\documentclass[11pt]{article}
\usepackage{amssymb,amsmath,amsthm,latexsym,fullpage,tabularx,color,bbm}
\usepackage{enumerate}
\usepackage{xr-hyper}
\usepackage{hyperref}
\usepackage{graphicx}
\usepackage{subcaption}
\usepackage[title]{appendix}
\usepackage{float}
\usepackage{cleveref}
\usepackage{colortbl}

\usepackage[margin=1.15in]{geometry}

\numberwithin{equation}{section}

\def\Rset{\mathbb{R}}

\def\rme{\mathrm{e}}
\def\rmd{\mathrm{d}}

\def\esp{\mathbb{E}}

\def\var{\mathrm{var}}

\def\tn{\tilde{n}}

\definecolor{Gray}{gray}{0.85}

\newcommand\1[1]{\mathbbm{1}_{#1}}

\newtheorem{theorem}{Theorem}[section]
\newtheorem{lemma}[theorem]{Lemma}

\newtheorem{corollary}[theorem]{Corollary}

\begin{document}

{\renewcommand{\baselinestretch}{1.5}

\begin{titlepage}
    \title{
     \vspace{0in}
   On the Expectation of the Local-to-Zero Cross-Validated Log Likelihood Criterion for Bandwidth Selection in Kernel Spectral Estimation
     \vspace{.05in}
    }
      
  \author{Meng-Chen Hsieh  ~ \ \ Clifford Hurvich \footnote{Department of
  Technology, Operations and Statistics, Stern School of Business, New York University, USA.}}
      
 \date{\normalsize \today}
 \end{titlepage}
}

\maketitle

\begin{abstract}
  We consider data-driven bandwidth selection for a kernel spectral estimator at zero frequency based on a local-to-zero version of the cross validated log-likelihood (CVLL) criterion. The modified version is $\mbox{CVLL}_c$, based on a sum over Fourier frequencies from $1$ to $n^c$ with $0<c<1$, where $n$ is the sample size. We focus on the expectation of a key term in a Taylor series expansion for $\mbox{CVLL}_c$ and show that in the case $4/5 < c < 1$ it converges to the corresponding asymptotic mean squared error of the spectral estimator at zero frequency. This provides some justification for the use of the local CVLL criterion for Heteroskedasticity and Autocorrelation Consistent (HAC) standard error estimation. Our theoretical results do not follow from existing literature on CVLL because those results exploit the fact that CVLL is global, summing over all frequencies in $(0,\pi)$ rather than local-to-zero frequency, as is the case for $\mbox{CVLL}_c$.
\end{abstract}

\newpage
We consider data-driven bandwidth selection for a kernel spectral estimator at zero frequency based on a local-to-zero version of the cross validated log-likelihood (CVLL) criterion. The modified version is $\mbox{CVLL}_c$, based on a sum over Fourier frequencies from $1$ to $n^c$ with $0<c<1$, where $n$ is the sample size. We focus on the expectation of a key term in a Taylor series expansion for $\mbox{CVLL}_c$ and show that in the case $4/5< c < 1$ it converges to the corresponding asymptotic mean squared error of the spectral estimator at zero frequency. This provides some justification for the use of the local CVLL criterion for Heteroskedasticity and Autocorrelation Consistent (HAC) standard error estimation. Our theoretical results do not follow from existing literature on CVLL because those results exploit the fact that CVLL is global, summing over all frequencies in $(0,\pi)$ rather than local-to-zero frequency, as is the case for $\mbox{CVLL}_c$.

Assume that $\{X_t\}$ is a weakly stationary short-memory time series with positive spectral density $f(\omega)$, $\omega \in [-\pi, \pi]$.  We represent $\{X_t\}$ as
\begin{equation}\label{eq:def_xt}
X_t = \mu + \sum_{j=0}^{\infty}\beta_j \epsilon_{t-j} 
\end{equation}
with $\sum_{j=0}^{\infty}j^{\frac{1}{2}}|\beta_j| < \infty$, where $\{\epsilon_t\}$ is iid with mean zero, variance $\sigma^2_{\epsilon}=1$, and finite fourth moment. Suppose further that $\{X_t\}$ has lag-$r$ autocovariance $\gamma_r$ with $\sum_{r=-\infty}^{\infty}r^2 |\gamma_r| < \infty$.
The long run variance of the sample mean $\bar{X}_n = \frac{1}{n}\sum_{t=1}^n X_t$ is
\begin{equation}\label{eq:stand_error_freq}
\lim_{n \rightarrow \infty}  n \; \var\left(\bar{X}_n\right) = \sum_{r=-\infty}^\infty \gamma_r = 2\pi f(0) \;.
\end{equation}
Let $f(\omega_j)$ be the spectral density of $\{X_t\}$ at the fourier frequency $\omega_j=\frac{2 \pi j}{n}$, $j=0,1,2,\ldots$.
Define the periodogram by $I_j=\frac{2\pi}{n}|\sum_{t=1}^n x_t \exp(-i\omega_j t)|^2$. Note that we do not correct for the sample mean in the definition of $I_j$. Mean correction
has no effect on $I_j$ for $j \ne 0$.
Denote by $\hat{f}_j$ the discrete periodogram average estimate of $f(\omega_j)$
\begin{equation}\label{eq:def_ave_pgrm_spec_est}
 \hat{f}_j = \hat{f}_m(\omega_j)=\frac{2\pi}{n}\sum_{\ell=1}^{n-1}K_m\left(\omega_j - \omega_{\ell} \right)I(\omega_{\ell})\;,
\end{equation}
where $K_m(\omega)=m\sum_{j=-\infty}^{\infty}K(m(\omega+2\pi j))$. Since the estimator does not use the periodogram at zero frequency, the estimator is invariant to the addition of a constant to the time series. We therefore henceforth assume without loss of generality that $\mu=0$ and we do not correct sample autocovariances for the sample mean. 
The function $K$ is the \textit{spectral window} which satisfies
\[
\int_{\Rset}|K(\omega)| \rmd \omega < \infty \notag \;, \;\; \int_{\Rset}K(\omega) \rmd \omega = 1 \;.
\]
The Fourier inverse of $K$ is called the \textit{lag window} or kernel
\begin{equation}
k(x)=\int_{\Rset}K(\omega)\rme^{ix\omega}\rmd \omega < \infty \;, \; x\in \Rset
\end{equation}
which is a key component in the widely used \textit{lag-weights} estimate of $f(\omega_j)$  
\begin{equation}\label{eq:lag_weights_est}
\tilde{f}_j = \frac{1}{2\pi}\sum_{|r|\leq m}k\left(\frac{|r|}{m}\right)\hat{\gamma}_r \rme^{ir\omega_j} \;,
\end{equation}
where $\hat{\gamma}_r$ is the sample lag-$r$ autocovariance, $\hat{\gamma}_r=\frac{1}{n}\sum_{t=|r|+1}^n x_t x_{t-|r|}$. We assume that the lag window function $k(x)$ is even, satisfying  $k(0)=1$, $k(x) \leq 1$ if $|x|\leq 1$; $k(x)=0$, if $|x|>1$.  It follows that for any lag $r$ with $|r| > m$, $k(r/m) = 0$ and $k(r/m) \leq 1$,  if $|r| \leq  m$. Here, $m$ is called the \textit{truncation point} or the (inverse of the) \textit{bandwidth}. For simplicity and definiteness, we will henceforth assume that $k(x)$ follows Parzen's lag window defined by Eq.(\ref{eq:Parzen_lag_window}), which has a finite truncation point $m$.

%Following Robinson's (1991), we apply the following formula to construct the \textbf{leave-one-out} spectral density estimate of $f(\omega_j)$ from a sample time series $x_1, x_2, \ldots, x_n$
%\begin{equation}\label{eq:def_LOO_spec_est}
% \hat{f}^{m}_{(j)}(\omega_j) = \frac{2\pi}{n}\sum_{\substack{\ell=1 \\ \ell \neq j, n-j}}^{n}K_m\left(\omega_j - \omega_{\ell} \right)I(\omega_{\ell})\;,
%\end{equation}
%where $I(\omega_{\ell})$ is the periodogram at $\omega_{\ell}$. The subscript $(j)$ in $\hat{f}^{m}_{(j)}(\omega_j)$ indicates that all frequencies except for $\omega_j$ and $\omega_{n-j}$ are utilized to estimate $f(\omega_j)$. 
%The fourier inverse of $K$ is called the \textbf{lag window} or kernel
%\begin{equation}
%k(x)=\int_{\Rset}K(\omega)\rme^{ix\omega}\rmd \omega < \infty \;, \; x\in \Rset
%\end{equation}
%which is a key component in the widely used \textbf{lag-weights} estimate of $f(\omega_j)$  
%\begin{equation}\label{eq:lag_weights_est}
%\hat{f}(\omega_j) = \frac{1}{2\pi}\sum_{|r|\leq m}k\left(\frac{|r|}{m}\right)\hat{c}_r \rme^{ir\omega_j} \;,
%\end{equation}
%where $\hat{c}_r$ is the sample lag-r autocovariance of time series $\{x_t\}$. The lag window function $k(x)$ is even, satisfying  $k(0)=1$, $k(x) \leq 1$ if $|x|\leq 1$; $k(x)=0$, if $|x|>1$.  It follows that for any lag number $|r| > m$, $k(r/m) = 0$ and $k(r/m) \leq 1$,  if $|r| \leq  m$. Hence $m$ is called the truncation point or the bandwidth in the nonparametric spectrum estimate literature.  

To select the optimal $m$, one may wish to minimize the mean integrated squared percentage error (MISPE)
\begin{equation}\label{eq:MISLE_discrete}
\mbox{MISPE}(m) = \esp \left[\frac{1}{n^{\prime}}\sum_{j=1}^{n^{\prime}} \left( \frac{\hat{f}_m(\omega_j)-f(\omega_j)}{f(\omega_j)}\right)^2  \right] \;,
\end{equation}
where $n^{\prime}=\lfloor \frac{n-1}{2}  \rfloor$.
However, this criterion is not feasible as it depends on the unknown spectral density $f(\omega_j)$. 
Beltrao and Bloomfield (1987) introduced a leave-one-out cross validation version of Whittle’s approximation to minus twice the Gaussian log-likelihood 
\begin{equation}\label{eq:CVLL_BB}
\mbox{CVLL}(m) = \frac{1}{n^{\prime}}\sum_{j=1}^{n^{\prime}}\left[\hat{f}^m_{(j)}(\omega_j)+ \frac{I(\omega_j)}{\hat{f}^m_{(j)}(\omega_j)}  \right] \;,
\end{equation}
and they further showed that for a given $m$, $\mbox{CVLL}(m)$ is asymptotically equivalent to $\mbox{MISPE}(m)$. Their approach is feasible since $\hat{f}^m_{(j)}$ is the leave-one-out estimate of $f(\omega_j)$, which can be directly computed from the data. Robinson (1991) proved the stronger result that the minimizer of $\mbox{CVLL}$ is asymptotically equivalent to the minimizer of $\mbox{MISPE}$. We follow Robinson (1991) to construct the leave-one-out spectral density estimate of $f(\omega_j)$ from a sample $x_1, x_2, \ldots, x_n$
\begin{equation}\label{eq:def_LOO_spec_est}
 \hat{f}^{m}_{(j)}(\omega_j) = \frac{2\pi}{n}\sum_{\substack{\ell=1 \\ \ell \neq j, n-j}}^{n-1}K_m\left(\omega_j - \omega_{\ell} \right)I(\omega_{\ell})\;,
\end{equation}
where $I(\omega_{\ell})=I_{\ell}$ is the periodogram at $\omega_{\ell}$. The subscript $(j)$ in $\hat{f}^{m}_{(j)}(\omega_j)$ indicates that all frequencies except for $\omega_j$ and $\omega_{n-j}$ are utilized to estimate $f(\omega_j)$. 

In view of Eq.(\ref{eq:stand_error_freq}) and related formulas, HAC (heteroskedasticity and autocorrelation consistent) standard error estimation focuses on the estimation of $f(0)$. Given this focus on low frequencies, Xu and Hurvich (2026) proposed a local-to-zero cross-validated log likelihood for estimating $f(0)$ 
\begin{equation}\label{eq:CVLL_H}
\mbox{CVLL}_c(m) =  \frac{1}{\tilde{n}}\sum_{j=1}^{\tilde{n}}\left[\hat{f}^m_{(j)}(\omega_j)+ \frac{I(\omega_j)}{\hat{f}^m_{(j)}(\omega_j)}  \right]\;,
\end{equation}
where $\tilde{n} = n^c$, $0 < c < 1$. Since the frequencies used in $\mbox{CVLL}_c(m)$ all lie in a neighborhood that shrinks to zero, it can be hoped that minimizing $\mbox{CVLL}_c(m)$ is asymptotically equivalent to minimizing 
\begin{equation}
\esp \left[\left( \frac{\hat{f}_m(0)-f(0)}{f(0)}\right)^2  \right] \;.
\end{equation}    

For simplicity, denote $f_j = f(\omega_j)$, $I_j = I(\omega_j)$, $\hat{f}^j_j = \hat{f}^{m}_{(j)}(\omega_j)$ the leave-one-out estimate (see Eq.(\ref{eq:def_LOO_spec_est})).
If $f$ is twice continuously differentiable at zero frequency (as we assume here), the asymptotically optimal value of $m$ is proportional to $n^{1/5}$. Thus, we define $\tau$
by $m=\tau n^{1/5}$ and focus henceforth on selection of $\tau$. In the sequel for brevity, we will write $\mbox{CVLL}_c(\tau)$ instead of $\mbox{CVLL}_c(\tau n^{\frac{1}{5}})$.
It follows from Parzen (1957) that 
\begin{equation}
n^{\frac{4}{5}}\;\esp \left[\left( \frac{\hat{f}_{\tau n^{\frac{1}{5}}}(0)-f(0)}{f(0)}\right)^2  \right] \rightarrow c_0(\tau)\;,
\end{equation}   
where 
\begin{equation}\label{def:co_ttau}
c_o(\tau) = \frac{1}{2}\left\{\tau\kappa + \tau^{-4}\left(h^{(2)}\frac{f^{(2)}(0)}{f(0)}\right)^2 \right\}\;,
\end{equation}
$\kappa=\int_{\Rset}k^2(x)\rmd x$, $k(x)$ is Parzen's lag window function defined by Eq.(\ref{eq:Parzen_lag_window}), 
$f^{(2)}(0)$ is the second spectral derivative evaluated at $0$,
and $h^{(2)}$ is the characteristic coefficient defined by Eq.(\ref{def:char_coef}). The two additive terms in Eq.(\ref{def:co_ttau}) may be viewed as asymptotic variance and squared bias, respectively. 

Define
\begin{eqnarray}
\tilde{L} = \tilde{L}(f) &=& \sum_{j=1}^{\tn} \left(\log f_j + \frac{I_j}{f_j}  \right)\;, \\
\tilde{L}(\tau) &=& \sum_{j=1}^{\tn} \left(\log \hat{f}^j_j + \frac{I_j}{\hat{f}^j_j}  \right) = \tilde{n}\; \mbox{CVLL}_c(\tau)\;.
\end{eqnarray}
\\
Note that $\tilde{L}$ does not depend on $\tau$. By Taylor series expansion, 
\begin{equation}\label{eq:dif_L_mod}
\tilde{L}(\tau) - \tilde{L} =\sum_{j=1}^{\tn}\left(\frac{\hat{f}^j_j}{f_j}-1\right)\left(1-\frac{I_j}{f_j}\right)
           -\sum_{j=1}^{\tn} \left(\frac{\hat{f}^j_j}{f_j}-1\right)^2\left(1-\frac{I_j}{f_j} \right)
           + \frac{1}{2}\sum_{j=1}^{\tn} \left(\frac{\hat{f}^j_j}{f_j}-1 \right)^2
           + O_p\left(\tn\left(\frac{\hat{f}^j_j}{f_j}-1 \right)^3 \right)\;.
\end{equation}
Hence
\begin{eqnarray}
n^{4/5-c}\left[\tilde{L}(\tau) - \tilde{L}\right]
&=& n^{4/5-c}\left[\sum_{j=1}^{\tn}\left(\frac{\hat{f}^j_j}{f_j}-1\right)\left(\frac{I_j}{f_j}-1 \right)\right] \label{eq:likhood_diff_1}\\
&-& n^{4/5-c}\left[\sum_{j=1}^{\tn}\left(\frac{I_j}{f_j}-1 \right)\left(\frac{\hat{f}^j_j}{f_j}-1\right)^2 \right] \label{eq:likhood_diff_2}\\
&+& \frac{1}{2}\left[n^{4/5-c}\sum_{j=1}^{\tn}\left(\frac{\hat{f}^j_j}{f_j}-1\right)^2 \right] \label{eq:likhood_diff_3}\\
&+& O_p\left(\tn\left(\frac{\hat{f}^j_j}{f_j}-1 \right)^3 \right)\label{eq:likhood_diff_4}
\end{eqnarray}
\\
Using $\max_{\tau,j}|\hat{f}^j_j - f_j| = O_p(n^{-2/5})$ from Robinson (1991)'s (C.8), it follows that the last term
\[
O_p\left(\tn\left(\frac{\hat{f}^j_j}{f_j}-1 \right)^3 \right) = O_p(n^{c-6/5}) = o_p(1)\;,
\]
since $0 < c < 1$.
\\
The asymptotic properties of $\tilde{L} (\tau) - \tilde L$ do not follow from, and cannot be easily derived from those for the non-local CVLL established by Robinson (1991). One key reason is that Robinson (1991) relied in a crucial way on Parseval's formula (see bottom of page 1357 of Robinson (1991), and compare with our Eq.(\ref{eq:norm_esp_gjSqr}) below), and Parseval's formula does not hold for sums that are local to zero frequency. For simplicity, we set the more modest goal of examining the expectation, $E[\tilde L (\tau) - \tilde L]$ and whether it behaves like $c_0 (\tau)$. More specifically, we will focus only on the expectation of (\ref{eq:likhood_diff_3}). Some justification for ignoring the expectation of (\ref{eq:likhood_diff_4}) was provided above. In the extremely special case where $\{X_t\}$ is Gaussian white noise, the expectation of (\ref{eq:likhood_diff_1}) and (\ref{eq:likhood_diff_2}) is exactly zero, due to the independence of $\{I_j\}$ and the use of the leave-one-out estimate. We believe that it could be proved that the expectations of (\ref{eq:likhood_diff_1}) and (\ref{eq:likhood_diff_2}) are negligible more generally under suitable regularity conditions, but we do not pursue this here.
We have the following theorem: 
\begin{theorem}
Let $m = \tau n^{1/5}$ and $\tilde{n}=n^c$,  where $\tau> 0$, and $4/5 <c <1$. 
As $n \to \infty$, 
\label{prop:exp_CVLL_convergence}
\begin{equation}\label{eq:true_CVLL_H} 
\frac{1}{2}\left[n^{4/5-c}\sum_{j=1}^{\tn}\esp\left(\frac{\hat{f}^j_j}{f_j}-1\right)^2 \right] \to c_0(\tau) \;.
\end{equation}
\end{theorem}
Thus, for each fixed $\tau$, the expectation of Eq.(\ref{eq:likhood_diff_3}) appearing in the Taylor series expansion for $\mbox{CVLL}_c(\tau)$ converges to the asymptotic mean squared percentage error of the discrete periodogram average estimate of $f(0)$. 

\begin{proof}
Let $\hat{f}_j$ be the discrete periodogram average estimate of $f(\omega_j)$ given by Eq.(\ref{eq:def_ave_pgrm_spec_est})
%\begin{equation}\label{eq:def_LOO_spec_est}
% \hat{f}_j = \frac{2\pi}{n}\sum_{\ell=1}^{n-1}K_m\left(\omega_j - \omega_{\ell} \right)I(\omega_{\ell})\;,
%\end{equation}
and $\tilde{f}_j$ the corresponding lag-weights estimate of $f(\omega_j)$ (See Eq.(\ref{eq:lag_weights_est})).
Let $g_j$ be the lag-weights spectral density estimate for $\{\epsilon_t\}$ at $\omega_j$, 
\begin{equation}\label{def:g_j}
g_j = \frac{1}{2\pi}\sum_{|r| \leq m} k(r/m)\rme^{ir\lambda_j}\hat{\gamma}_{\epsilon_r} \;, 
\end{equation}
where $\hat{\gamma}_{\epsilon_r}=\frac{1}{n}\sum_{t=|r|+1}^n \epsilon_t\epsilon_{t-|r|}$. 
Define $g^{\ast}_j = 2\pi g_j-1$. Then $\esp[2 \pi g_j] = 1$ and $\esp[g_j^{\ast}] = 0$.
\\
We can express
\begin{eqnarray}\label{eq:decom_ffj_ratio}
\frac{\hat{f}^j_j}{f_j}-1 &=& (\hat{f}^j_j-\tilde{f}_j)/f_j + \left\{(\tilde{f}_j-\esp \tilde{f}_j)/f_j - (2\pi g_j-1) \right\} + \left[(2\pi g_j-1) + \esp \tilde{f}_j/f_j -1 \right] \notag \\
&=& (\hat{f}^j_j-\tilde{f}_j)/f_j + \left\{(\tilde{f}_j-\esp \tilde{f}_j)/f_j - g^{\ast}_j \right\} + (g^{\ast}_j + \esp \tilde{f}_j/f_j -1 )\;. \notag \\
\end{eqnarray}
Since $\max_{\tau, j}|\hat{f}^j_j - \tilde{f}_j| = O_p(n^{-3/5})$ and $\max_{\tau,j}\left|(\tilde{f}_j-\esp \tilde{f}_j)/f_j -g^{\ast}_j \right|= O_p(n^{-2/5}/\log n)$ (See Robinson (1991)'s (C.7) and (C.9)). It follows that the first two terms on the RHS of Eq.(\ref{eq:decom_ffj_ratio}) are negligible. 
Therefore, to prove Eq.(\ref{eq:true_CVLL_H}), it suffices to show that
\begin{equation}
{n}^{4/5-c}\sum_{j=1}^{\tn}\esp\left\{g^{\ast}_j + \esp \tilde{f}_j/f_j -1 \right\}^2 - 2 c_0(\tau) \rightarrow 0 \;.
\end{equation}
\\
Expanding the squared term, we obtain
\begin{eqnarray}
&&{n}^{4/5-c}\sum_{j=1}^{\tn}\esp \left\{g^{\ast}_j + \esp \tilde{f}_j/f_j -1 \right\}^2  \notag \\
&=&  {n}^{4/5-c}\sum_{j=1}^{\tn} \esp\left[{g^{\ast}_j}^2\right] \\
& + & 2 {n}^{4/5-c}\sum_{j=1}^{\tn} \esp\left[g^{\ast}_j \right]\left(\esp \tilde{f}_j/f_j -1\right) \\
&+&  {n}^{4/5-c}\sum_{j=1}^{\tn}\esp\left(\esp \tilde{f}_j/f_j -1\right)^2
\end{eqnarray}
Since $\esp[g^{\ast}_j]=0$, by Lemma \ref{lem:convg_bias_square} and Lemma \ref{lem:convg_exp_g}, it follows that
\begin{equation}
{n}^{4/5-c}\sum_{j=1}^{\tn}\esp \left\{g^{\ast}_j + \esp \tilde{f}_j/f_j -1 \right\}^2  \rightarrow \tau \kappa + b_0(\tau)\;.
\end{equation} 

\end{proof}

\begin{lemma}\label{lem:convg_bias_square}
Assume that $\tau \in [u,v]$, $ 0 < u < v < \infty$ and $ 0<c < 1$. Then
\begin{equation}\label{eq:3rd_conveg}
\max_{\tau}\left|{n}^{4/5-c}\sum_{j=1}^{\tn}\left(\esp[\tilde{f}(\omega_j)]/f(\omega_j) -1\right)^2 - b_0(\tau)\right| \stackrel{}{\rightarrow} 0 \;,
\end{equation}
\noindent where
\begin{equation}
b_0(\tau) = \tau^{-4}\left(h^{(2)}\frac{f^{(2)}(\omega_0)}{f(\omega_0)}\right)^2  \;,
\end{equation}

\noindent $h^{(2)}$ is the characteristic coefficient defined by
\begin{equation}\label{def:char_coef}
h^{(2)} = \lim_{z \to 0}\frac{1-k(z)}{|z|^2}\;,
\end{equation}

\noindent $k(z)$ is Parzen's lag window, $f^{(2)}(\omega_0)$ is the second spectral derivative evaluated at $\omega_0$,

\begin{equation}
 f^{(2)}(\omega_0) = \sum_{r=-\infty}^{\infty}|r|^{2}\gamma_r\;.
\end{equation}

\end{lemma}

\begin{proof}
Define Parzen's lag window function $k(z)$ by
\begin{eqnarray}\label{eq:Parzen_lag_window}
k(z) &=& 1-6z^2(1-|z|)\;, \;\;\;\;  0 \leq |z| \leq 1/2 \notag \\
     &=& 2(1-|z|)^3\;, \;\;\;\;\;\; \; 1/2 < |z| \leq 1 \notag \\
     &=& 0\;, \;\;\;\;\;\; \; |z| \geq 1 \;.
\end{eqnarray}
Hence
\begin{equation}\label{def:char_coef}
h^{(2)} = \lim_{z \to 0}\frac{1-k(z)}{|z|^2} = 6 \;.
\end{equation}

\noindent We can express
\begin{eqnarray}
\left|n^{4/5-c}\sum_{j=1}^{\tn}\left(\esp \tilde{f}(\omega_j)/f(\omega_j)-1 \right)^2 -b_0(\tau) \right|
%&=& \left|n^{4/5}\frac{1}{\tn} \sum_{j=1}^{\tn}\left(\esp \tilde{f}(\omega_j)/f(\omega_j)-1\right)^2 -b_0(\tau) \right| \notag \\
&=& \left| \frac{1}{\tn} \sum_{j=1}^{\tn}\left[n^{4/5}\left(\esp \tilde{f}(\omega_j)/f(\omega_j)-1\right)^2 -b_0(\tau) \right] \right| \notag \\
& \leq & \frac{1}{\tn}\sum_{j=1}^{\tn}\left|n^{4/5}\left(\esp \tilde{f}(\omega_j)/f(\omega_j)-1\right)^2 -b_0(\tau) \right| \notag \\
& \leq & \frac{1}{\tn}  \left(\tn\; \max_{j \in (1, \tn)}\left|n^{4/5}\left(\esp \tilde{f}(\omega_j)/f(\omega_j)-1\right)^2 -b_0(\tau) \right| \right)  \notag \\
&=& \max_{j \in (1, \tn)}\left|n^{4/5}\left(\esp \tilde{f}(\omega_j)/f(\omega_j)-1\right)^2 -b_0(\tau) \right|
\end{eqnarray}

\noindent It follows that
\begin{eqnarray}
&&\max_{\tau}\left|n^{4/5-c}\sum_{j=1}^{\tn}\left(\esp \tilde{f}(\omega_j)/f(\omega_j)-1 \right)^2 -b_0(\tau) \right| \notag \\
& \leq &  \max_{\tau,\; j \in (1, \tn)}\left|n^{4/5}\left(\esp \tilde{f}(\omega_j)/f(\omega_j)-1\right)^2 -b_0(\tau) \right| \notag \\
&=& \max_{\tau,\; j \in (1, \tn)}\left|n^{4/5}\left(\esp \tilde{f}(\omega_j)/f(\omega_j)-1\right)^2 -\tau^{-4}[h^{(2)}f^{(2)}(\omega_j)/f(\omega_j)]^2  \right| \label{eq:covg_b0_tau1} \\
&+& \max_{\tau,\; j \in (1, \tn)}\left| \tau^{-4}(h^{(2)}f^{(2)}(\omega_j)/f(\omega_j)^2 -b_0(\tau)  \right| \;, \label{eq:convg_b0_tau2}
\end{eqnarray}
where $f^{(2)}(\omega_j)$ is the second spectral derivative defined by
\begin{equation}
f^{(2)}(\omega_j) = \frac{1}{2\pi}\sum_{r=-\infty}^{\infty}r^2 \gamma_r \rme^{ir\omega_j}\;.
\end{equation}

\noindent First, we look at (\ref{eq:convg_b0_tau2}). Since $\tau \in (u,v)$, and $0 < u < v < \infty$,
\begin{eqnarray}
\max_{\tau,\; j \in (1, \tn)}\left| \tau^{-4}(h^{(2)}f^{(2)}(\omega_j)/f(\omega_j)^2 -b_0(\tau)  \right|
%&=& \max_{\tau,\; j \in (1, \tn)}\left| \tau^{-4}[h^{(2)}f^{(2)}(\omega_j)/f(\omega_j)]^2 - \tau^{-4}[h^{(2)}f^{(2)}(\omega_0)/f(\omega_0)]^2  \right| \notag \\
&=& \max_{\tau,\; j \in (1, \tn)} \tau^{-4}(h^{(2)})^2 \left|[f^{(2)}(\omega_j)/f(\omega_j)]^2 - [f^{(2)}(\omega_0)/f(\omega_0)]^2  \right| \notag \\
& \leq & u^{-4}\max_{j \in (1, \tn)}(h^{(2)})^2 \left|[f^{(2)}(\omega_j)/f(\omega_j)]^2 - [f^{(2)}(\omega_0)/f(\omega_0)]^2  \right| \notag \\
& = & u^{-4}\max_{j \in (1, \tn)}(h^{(2)})^2 \left|\frac{f^{(2)}(\omega_j)}{f(\omega_j)}+\frac{f^{(2)}(\omega_0)}{f(\omega_0)}\right|\left|\frac{f^{(2)}(\omega_j)}{f(\omega_j)} - \frac{f^{(2)}(\omega_0)}{f(\omega_0)}  \right| \notag\\
\end{eqnarray}
where
\begin{eqnarray}
\left|\frac{f^{(2)}(\omega_j)}{f(\omega_j)} - \frac{f^{(2)}(\omega_0)}{f(\omega_0)}  \right|
&\leq & \left| \frac{f^{(2)}(\omega_j)}{f(\omega_j)} - \frac{f^{(2)}(\omega_j)}{f(\omega_0)} \right|
+  \left|\frac{f^{(2)}(\omega_j)}{f(\omega_0)} - \frac{f^{(2)}(\omega_0)}{f(\omega_0)}  \right|  \notag \\
&=& \frac{1}{f(\omega_j)f(\omega_0)}\left| f^{(2)}(\omega_j)\big[f(\omega_0) - f(\omega_j) \big]\right|
 + \frac{1}{f(\omega_0)} \left|f^{(2)}(\omega_j) - f^{(2)}(\omega_0)  \right| \notag \\
\end{eqnarray}
$\forall \; 1 \leq j \leq \tn$, as $ n \to \infty$, $\omega_j \to \omega_0$.  By the continuity of $f$ and $f^{(2)}$, $f(\omega_j) \to f(\omega_0)$, and $f^{(2)}(\omega_j) \to f^{(2)}(\omega_0)$ uniformly in $j$, $\forall 1 \leq j \leq \tn$.

\noindent It follows that
\begin{equation}
\max_{j \in [1, \tn]}\left|\frac{f^{(2)}(\omega_j)}{f(\omega_j)} - \frac{f^{(2)}(\omega_0)}{f(\omega_0)}  \right|  \to 0\;.
\end{equation}
Since $\max_{j \in [1,\tn]}\left|\frac{f^{(2)}(\omega_j)}{f(\omega_j)}+\frac{f^{(2)}(\omega_0)}{f(\omega_0)}\right|$ is bounded, we prove
\[
\max_{\tau,\; j \in (1, \tn)}\left| \tau^{-4}(h^{(2)}f^{(2)}(\omega_j)/f(\omega_j)^2 -b_0(\tau)  \right| \to 0\;.
\]

\noindent Next we consider (\ref{eq:covg_b0_tau1}), which can be represented as
\begin{eqnarray}\label{eq:decomp_sq_diff}
&& \max_{\tau,\; j \in (1, \tn)}\left|n^{4/5}\left[\esp \tilde{f}(\omega_j)/f(\omega_j)-1\right]^2 -\tau^{-4}\left[h^{(2)}f^{2}(\omega_j)/f(\omega_j)^2\right]  \right| \notag \\
&=& \max_{\tau,\; j \in (1, \tn)} \frac{1}{f(\omega_j)^2}\left|n^{4/5}\left[\esp \tilde{f}(\omega_j)-f(\omega_j)\right]^2 - \tau^{-4}\left[h^{(2)}f^{(2)}(\omega_j)\right]^2   \right| \notag \\
&=& \max_{\tau,\; j \in (1, \tn)} \frac{1}{f(\omega_j)^2}\left|\left[n^{2/5}\left(\esp \tilde{f}(\omega_j)-f(\omega_j)\right)+\tau^{-2}h^{(2)}f^{(2)}(\omega_j)\right]
\left[n^{2/5}\left(\esp \tilde{f}(\omega_j)-f(\omega_j)\right)-\tau^{-2}h^{(2)}f^{(2)}(\omega_j)\right]\right| \notag \\
&=& \max_{\tau,\; j \in (1, \tn)} \frac{1}{f(\omega_j)^2\tau^2}\left|\tau^2 n^{2/5}\left(\esp \tilde{f}(\omega_j)-f(\omega_j)\right)+h^{(2)}f^{(2)}(\omega_j)\right|
\left|n^{2/5}\left(\esp \tilde{f}(\omega_j)-f(\omega_j)\right)-\tau^{-2}h^{(2)}f^{(2)}(\omega_j)\right| \notag \\ 
\end{eqnarray}
Note that $\esp[\hat{\gamma}_r] = \left[\frac{n-|r|}{n}\right]{\gamma}_r$.
We can then express
\begin{eqnarray}
\esp [\tilde{f}(\omega_j)]-f(\omega_j) &=& \frac{1}{2\pi}\sum_{|r|\leq m}k(r/m)\left(\frac{n-|r|}{n} \right){\gamma}_r \rme^{ir\omega_j} - \frac{1}{2\pi}\sum_{r=-\infty}^{\infty}{\gamma}_r\rme^{ir\omega_j} \notag \\
&=& \frac{1}{2\pi}\sum_{|r|\leq m}\left\{k(r/m)-\frac{|r|}{n}k(r/m)-1\right\}{\gamma}_r \rme^{ir\omega_j} -\frac{1}{2\pi}\sum_{|r|> m}{\gamma}_r\rme^{ir\omega_j} \notag \\
&=& \frac{1}{2\pi}\sum_{|r|\leq m}\left[k(r/m)-1 \right]\gamma_r \rme^{ir\omega_j}  \\
&-& \frac{1}{2\pi}\sum_{|r|\leq m}\left(\frac{|r|}{n}\right)k(r/m)\gamma_r \rme^{ir\omega_j}  \\
&-&  \frac{1}{2\pi}\sum_{|r| > m}\gamma_r \rme^{ir\omega_j}
\end{eqnarray}

\noindent Multiply by $m^2$ on both sides of the above equation, we obtain
\begin{eqnarray}
m^2 \left( \esp [\tilde{f}(\omega_j)]-f(\omega_j) \right) &=& \tau^2 n^{2/5}\left(\esp [\tilde{f}(\omega_j)]-f(\omega_j)\right) \notag \\
%&=& \frac{1}{2\pi}\sum_{|r|\leq m}\left[\frac{k(r/m)-1}{\left(\frac{1}{m}\right)^2}\right]{\gamma}_r \rme^{ir\omega_j} \\
%&-& \frac{1}{2\pi}\sum_{|r|\leq m}\left(\frac{m^2}{n}\right)|r|k(r/m)\rme^{ir\omega_j} \\
%&-&  \frac{1}{2\pi}\sum_{|r| > m} m^2 \gamma_r \rme^{ir\omega_j} \\
&=&  \frac{1}{2\pi}\sum_{|r|\leq m}\left[\frac{k(r/m)-1}{\left(\frac{r}{m}\right)^2}\right]r^2{\gamma}_r \rme^{ir\omega_j} \label{eq:m2bias_1}\\
&-& \frac{1}{2\pi}\left(\frac{m^2}{n}\right)\sum_{|r|\leq m}|r|k(r/m)\rme^{ir\omega_j} \label{eq:m2bias_2}\\
&-&  \frac{1}{2\pi}\sum_{|r| > m} \left(\frac{m}{r}\right)^2 r^2 \gamma_r \rme^{ir\omega_j}\label{eq:m2bias_3}
\end{eqnarray}
Hence
\begin{eqnarray}
\tau^2 n^{2/5} \left( \esp [\tilde{f}(\omega_j)]-f(\omega_j) \right) + h^{(2)}f^{2}(\omega_j)
&=& \tau^2 n^{2/5} \left( \esp [\tilde{f}(\omega_j)]-f(\omega_j) \right) + h^{(2)}\frac{1}{2\pi}\sum_{r=-\infty}^{\infty}r^2{\gamma}_r \rme^{ir\omega_j} \notag \\
&=& \frac{1}{2\pi}\sum_{|r|\leq m}\left[\frac{k(r/m)-1}{\left(\frac{r}{m}\right)^2}+h^{(2)}\right]r^2{\gamma}_r \rme^{ir\omega_j}\label{eq:decom_bias_1} \\
&-& \frac{1}{2\pi}\left(\frac{m^2}{n}\right)\sum_{|r|\leq m}|r|k(r/m)\rme^{ir\omega_j} \label{eq:decom_bias_2}  \\
&+&  \frac{1}{2\pi}\sum_{|r| > m}\left[ h^{(2)} - \left(\frac{m}{r}\right)^2 \right] r^2 \gamma_r \rme^{ir\omega_j}\label{eq:decom_bias_3}
\end{eqnarray}

\noindent Consider (\ref{eq:decom_bias_1}):
\begin{eqnarray}
\max_{\tau,\; j \in [1, \tn]}\left|\frac{1}{2\pi}\sum_{|r|\leq m}\left[\frac{k(r/m)-1}{\left(\frac{r}{m}\right)^2}+h^{(2)}\right]r^2{\gamma}_r \rme^{ir\omega_j}    \right|
& \leq & \max_{\tau,\; j \in [1, \tn]}\frac{1}{2\pi}\sum_{|r|\leq m}\left|\left[h^{(2)}-\frac{1-k(r/m)}{\left(\frac{r}{m}\right)^2}\right]r^2{\gamma}_r \rme^{ir\omega_j}    \right| \notag \\
& \leq & \max_{\tau}\frac{1}{2\pi}\sum_{|r|\leq m}\left|\left[h^{(2)}-\frac{1-k(r/m)}{\left(\frac{r}{m}\right)^2}\right]r^2{\gamma}_r \right| \notag \\
\end{eqnarray}

\noindent Let $\xi(m) = [k(r/m)-1]/(r/m)^2$.
Its partial derivative with respect to $m$ is equal to
\begin{equation}\label{eq:dev_k}
\frac{\partial \xi(m)}{\partial m} = 2(m/r^2)\left[k(r/m)-1 \right] - (m/r)^2\frac{\partial k(r/m)}{\partial m}(r/m^2)
\end{equation}
Since $0 \leq k(r/m) \leq 1$, the first term in (\ref{eq:dev_k}), $(m/r^2)\left[k(r/m)-1 \right] < 0$. Now consider $\frac{\partial k(r/m)}{\partial m}$.
Without loss of generality, assume $r > 0$. Assume that Parzen's lag window function is applied to estimate $f(\omega_j)$.
\begin{itemize}
\item [(1)] If $0 \leq r/m \leq 1/2$.  Then $k(r/m) = 1 - 6(r/m)^2[1-r/m]$.  Thus
\begin{equation}
\frac{\partial k(r/m)}{\partial m} = \frac{r^2}{m^3}\left[12-18\left(\frac{r}{m}\right)  \right]
\end{equation}
Since $r/m \leq 1/2$, $\frac{\partial k(r/m)}{\partial m}> 0$.  It follows that the second term in (\ref{eq:dev_k}) is also negative and thus
$ \frac{\partial \xi(m)}{\partial m} < 0$.
Also under the range $0 \leq r/m \leq 1/2$, it can be easily shown that $0 \leq (1-k(r/m))/(r/m)^2 \leq 6$, hence $h^{(2)} - (1-k(r/m))/(r/m)^2 \geq 0$.

\item[(2)] If $1/2 < r/m \leq 1$. Then $k(r/m) = 2[1-(r/m)]^3$. Thus
\begin{equation}
\frac{\partial k(r/m)}{\partial m} = 6\left[1-\left(\frac{r}{m}\right) \right]^2 \left(\frac{r}{m^2}\right) > 0 \;,
\end{equation}
and it follows that $ \frac{\partial \xi(m)}{\partial m} < 0$.
Under the range $1/2 < r/m \leq 1$, it can be shown that $3/4 \leq (1-k(r/m))/(r/m)^2 \leq 4$, and thus $h^{(2)} - (1-k(r/m))/(r/m)^2 \geq 0$.
\end{itemize}
By (1) and (2), we conclude that $h^{(2)} - (1-k(r/m))/(r/m)^2 \geq 0$, and that for fixed $n$ and $r$, $[k(r/m)-1]/(r/m)^2$ is monotone decreasing with $m$, and thus is monotone decreasing with $\tau$.
\\
Let $m_{\tau} = \tau n^{1/5}$, $\tau \in [u,v]$. It follows that
\begin{eqnarray}
\max_{\tau}\frac{1}{2\pi}\sum_{|r|\leq m}\left|\left[h^{(2)}-\frac{1-k(r/m)}{\left(\frac{r}{m}\right)^2} \right]r^2{\gamma}_r\right|
 &=& \frac{1}{2\pi}\sum_{|r|\leq m_u}\left[h^{(2)}-\frac{1-k(r/m_u)}{\left(\frac{r}{m_u}\right)^2} \right]\left|r^2{\gamma}_r\right| \;, \notag \\
 &+& \frac{1}{2\pi}\sum_{m_u < |r|\leq m}\left[h^{(2)}-\frac{1-k(r/m_u)}{\left(\frac{r}{m_u}\right)^2} \right]\left|r^2{\gamma}_r\right| \;. \notag \\
\end{eqnarray}

\noindent For fixed $\tau$ and $r \leq m_u$, the sequence $[h^{(2)}-\frac{1-k(r/m_u)}{(r/m_u)^2}]|\gamma_r r^2|$ is bounded by $h^{(2)}|\gamma_r r^2|$, where $\sum|\gamma_r|r^2 < \infty$, and it converges to zero as $n \to \infty$.
Hence by the dominated convergence theorem , we obtain
\begin{eqnarray}\label{eq:spec_est_bias_1}
\lim_{n\to \infty}\frac{1}{2\pi}\sum_{|r|\leq m_u}\left[h^{(2)}-\frac{1-k(r/m_u)}{\left(\frac{r}{m_u}\right)^2} \right]\left|r^2{\gamma}_r\right|
  &=& 0 \;,
\end{eqnarray}
$ \forall m_u < r \leq m$, the sequence $[h^{(2)}-\frac{1-k(r/m_u)}{(r/m_u)^2}]\gamma_r r^2$ is bounded by $h^{(2)}\gamma_r r^2$. As $n \to \infty$, $\sum_{|r|> m_u}r^2|\gamma_r| \to 0$.  Hence
\begin{eqnarray}\label{eq:spec_est_bias_1b}
\lim_{n\to \infty} \frac{1}{2\pi}\sum_{m_u < |r|\leq m}\left[h^{(2)}-\frac{1-k(r/m_u)}{\left(\frac{r}{m_u}\right)^2} \right]\left|r^2{\gamma}_r \right| &=& 0 \;.
\end{eqnarray}

\noindent Consider (\ref{eq:decom_bias_2}). Since $\sum_{|r|\leq m}|r k(r/m)| \leq \sum_{|r|\leq m}|r|$ ,
\begin{eqnarray}\label{eq:spec_est_bias_2}
\max_{\tau,\; j \in [1, \tn]}\left|\frac{1}{2\pi}\left(\frac{m^2}{n}\right)\sum_{|r|\leq m}|r|k(r/m)\rme^{ir\omega_j}\right| & \leq & \max_{\tau,\; j \in [1, \tn]} \tau^2 n^{-3/5}\sum_{|r|\leq m}\left|rk(r/m)\rme^{ir\omega_j} \right|  \notag \\
& \leq & v^2 n^{-3/5}\sum_{|r|\leq m}\left| r k(r/m)\right| \to 0\;.
\end{eqnarray}
\noindent Finally consider (\ref{eq:decom_bias_3}).
\begin{eqnarray}
\max_{\tau,\; j \in [1, \tn]}\left|\frac{1}{2\pi}\sum_{|r| > m} \left[h^{(2)}-\left(\frac{m}{r}\right)^2\right] r^2 \gamma_r \rme^{ir\omega_j}\right|
& \leq & \max_{\tau,\; j \in [1, \tn]}\frac{1}{2\pi}\sum_{|r| > m}\left|\left( h^{(2)}-\frac{m^2}{r^2} \right)r^2{\gamma}_r \rme^{ir\omega_j} \right| \notag \\
& \leq & \max_{\tau} \frac{1}{2\pi}\sum_{|r| > m} \left( h^{(2)}-\frac{m^2}{r^2} \right) \left| r^2{\gamma}_r \right| \notag \\
& \leq &  \frac{1}{2\pi}\sum_{|r| > m_u}\left( h^{(2)}-\frac{m_u^2}{r^2} \right)\left|   r^2{\gamma}_r \right|
\end{eqnarray}

\noindent Since $|r| > m_u$,  the sequence $h^{(2)}-m_u^2/r^2$ is bounded by $h^{(2)}$. $\sum_{r=0}^{\infty}r^2|\gamma_r| < \infty$ implies that $\sum_{|r|> m_u}|\gamma_r|r^2 \to 0$.  Therefore,
\begin{eqnarray}
\lim_{n \to \infty}\frac{1}{2\pi}\sum_{|r| > m_u} \left( h^{(2)}-\frac{m_u^2}{r^2} \right) \left| r^2{\gamma}_r \right|
 & = & 0 \;.
\end{eqnarray}
\noindent Hence
\begin{equation}\label{eq:spec_est_bias_3}
\max_{\tau,\; j \in [1, \tn]}\left|\frac{1}{2\pi}\sum_{|r| > m} \left[h^{(2)}-\left(\frac{m}{r}\right)^2\right] r^2 \gamma_r \rme^{ir\omega_j}\right| \to 0\;.
\end{equation}

\noindent By (\ref{eq:spec_est_bias_1}), (\ref{eq:spec_est_bias_1b}), (\ref{eq:spec_est_bias_2}), and (\ref{eq:spec_est_bias_3}), we obtain
\begin{eqnarray}\label{eq:convg_bias_sq_1}
&& \max_{\tau,\; j \in [1, \tn]}\tau^{-2}\left|\tau^2 n^{2/5}\left(\esp [\tilde{f}(\omega_j)]-f(\omega_j)\right)  +  h^{(2)}f^{(2)}(\omega_j)\right| \notag \\
& \leq & u^{-2 }\max_{\tau,\; j \in [1, \tn]}\left|\tau^2 n^{2/5}\left(\esp [\tilde{f}(\omega_j)]-f(\omega_j)\right)  +  h^{(2)}f^{(2)}(\omega_j)\right| \to 0\;.
\end{eqnarray}

\noindent Finally, we will show that $\max_{\tau,\; j \in[1,\tn]}\left|n^{2/5}\left(\esp \tilde{f}(\omega_j)-f(\omega_j)\right)-\tau^{-2}h^{(2)}f^{(2)}(\omega_j)\right| < \infty $.
\\
In view of (\ref{eq:m2bias_1}), (\ref{eq:m2bias_2}), and (\ref{eq:m2bias_3}),
\begin{eqnarray}
&& \max_{\tau,\; j \in [1, \tn]}\left|n^{2/5}\left(\esp \tilde{f}(\omega_j)-f(\omega_j)\right)-\tau^{-2}h^{(2)}f^{(2)}(\omega_j)\right| \notag \\
& \leq & \max_{\tau,\; j \in [1, \tn]}\left\{\left|n^{2/5}\left(\esp \tilde{f}(\omega_j)-f(\omega_j)\right)\right|  + \left|\tau^{-2}h^{(2)}f^{(2)}(\omega_j)\right| \right\}\notag \\
& \leq & \max_{\tau,\; j \in [1, \tn]} \left\{\frac{1}{2\pi}\sum_{|r|\leq m}\left|\left(\frac{k(r/m)-1}{\left(\frac{r}{m}\right)^2}\right)r^2{\gamma}_r \right|
+ \frac{1}{2\pi}\left(\frac{m^2}{n}\right)\sum_{|r|\leq m}\left|r|k(r/m)\right| \right.\notag \\
&+& \left. \frac{1}{2\pi}\sum_{|r| > m} \left| \left(\frac{m}{r}\right)^2 r^2 \gamma_r \right|
+ \left|\tau^{-2}h^{(2)}f^{(2)}(\omega_j) \right| \right\}
\end{eqnarray}

\noindent Since $(k(r/m)-1)/(r/m)^2$ is bounded and converges to $h^{(2)}$ and $\sum_{r=0}^{\infty}r^2|\gamma_r|< \infty$, and when $r > m$, $m^2/r^2$ converges to zero as $n \to \infty$, together with (\ref{eq:spec_est_bias_2}), we show that
\begin{eqnarray}\label{eq:convg_bias_sq_2}
\max_{\tau,\; j \in [1, \tn]}\left|n^{2/5}\left(\esp \tilde{f}(\omega_j)-f(\omega_j)\right)-\tau^{-2}h^{(2)}f^{(2)}(\omega_j)\right| < \infty \;.
\end{eqnarray}

\noindent By (\ref{eq:convg_bias_sq_1}), (\ref{eq:convg_bias_sq_2}) and $f(\omega_j) < \infty$, we show (\ref{eq:decomp_sq_diff}) converges to zero.  
By the convergence of  (\ref{eq:covg_b0_tau1}) and (\ref{eq:convg_b0_tau2}), we prove  
\[
\max_{\tau}\left|{n}^{4/5-c}\sum_{j=1}^{\tn}\left(\esp[\tilde{f}(\omega_j)]/f(\omega_j) -1\right)^2 - b_0(\tau)\right| \stackrel{}{\rightarrow} 0 \;.
\]

\end{proof}

\begin{lemma}\label{lem:Willa_Cliff}
Define $D_{\tn}(\omega_r) = \frac{\sin(\pi r(\tn/n))}{\tn\sin(\pi r /n)}$, where $\omega_r = 2\pi r/n$ and $\tilde{n} = n^c$, $0 < c < 1$.
Let $m = \tau n^{1/5}$. If $\omega_r \in (-\pi, \pi)$, then
\begin{eqnarray*}
\sum_{0 < |r| \leq m} |D_{\tn}(\omega_r)| &=& O(n^{1-c}\log n)
\end{eqnarray*}
\end{lemma}
\begin{proof}
By Lemma 0 in Chen and Hurvich (1998), there exists a finite constant $C^{\ast}$ such that for $\omega \in (-\pi, \pi)$,
$|D_{\tn}(\omega)| \leq C^{\ast} \min(1, {(\tn |\omega|)}^{-1})$.
Since $0 < |r| \leq m$, if $\omega_r \in (-\pi, \pi)$, then for $n \geq (2 \tau)^{5/4}$,
\begin{eqnarray}
\sum_{0 <  |r| \leq m}\left|D_{\tn}(\omega_r)\right|
 &\leq & 2 \sum_{0 < r \leq m} C^{\ast}\min\left(1,\; \frac{1}{\tn\omega_r} \right) \notag \\
 &\leq & 2 \sum_{0 <  r \leq m} C^{\ast}\left(\frac{1}{\tn\omega_r}\right) \notag \\
 & = &    n^{1-c}\frac{C^{\ast}}{\pi}\sum_{r=1}^{\lfloor m \rfloor} \frac{1}{r}
 \sim   n^{1-c} \log (n)
\end{eqnarray}

\end{proof}

\begin{corollary}\label{cor:conv_fft_lag_win}
If $\omega_r \in (-\pi, \pi)$ and $4/5 < c < 1$, then
\[
n^{-1/5}\sum_{0 < r \leq m}|D_{\tn}(\omega_r)| \to 0 \;.
\]
\end{corollary}

\begin{lemma}\label{lem:convg_exp_g}
Let $\kappa=\int_{\Rset}k^2(x)\rmd x$. Assume $\int_{\Rset}k^2(x)\rmd x < \infty$ and $\int_{\Rset}|x|k^2(x)\rmd x < \infty$.
\\
If $4/5 < c < 1$,
\begin{equation}\label{eq:convg_exp_g}
 \left|n^{4/5} \esp\left[\frac{1}{\tn}\sum_{j=1}^{\tn}{g^{\ast}_j}^2\right]-\tau \kappa \right| \rightarrow 0 \;,
\end{equation}
where $\tau$ is the tuning parameter for the bandwidth $m=\tau n^{1/5}$. 
\end{lemma}

\begin{proof}
Since $\esp[g^{\ast}_j]=0$ and $\esp[2\pi g_j] = 1$,  $\var(g^{\ast}_j)=\esp[{g_j^{\ast}}^2]=\esp[{(2\pi g_j-1)}^2] = \esp[{(2\pi g_j)}^2]-2\esp[2\pi g_j]+1=\esp[(2\pi g_j)^2]-1 $.
\\
By (\ref{def:g_j}), we can express
\begin{eqnarray*}
|2\pi g_j|^2
&=& \sum_r k\left(\frac{r}{m}\right)\rme^{ir\lambda_j} \hat{\gamma}_{\epsilon_r}\sum_s k\left(\frac{s}{m}\right)\rme^{-is\lambda_j} \hat{\gamma}_{\epsilon_s} \\
&=& \sum_r\sum_s  k\left(\frac{r}{m}\right) k\left(\frac{s}{m}\right) \hat{\gamma}_{\epsilon_r}\hat{\gamma}_{\epsilon_s}\rme^{i(r-s)\lambda_j} \;,
\end{eqnarray*}
and thus
\begin{eqnarray}\label{eq:pigj_sq}
\esp|2\pi g_j|^2 &=&  \sum_r\sum_s  k\left(\frac{r}{m}\right) k\left(\frac{s}{m}\right)\rme^{i(r-s)\lambda_j}\esp[\hat{\gamma}_{\epsilon_r}\hat{\gamma}_{\epsilon_s}]
\end{eqnarray}
Note that
\begin{eqnarray}
\esp[\hat{\gamma}_{\epsilon_r}\hat{\gamma}_{\epsilon_s}]
 &=& \esp \left[\left(\frac{1}{n}\sum_{t=|r|}^{n-1}\epsilon_t\epsilon_{t-|r|}\right)\left(\frac{1}{n}\sum_{k=|s|}^{n-1}\epsilon_k\epsilon_{k-|s|}\right)\right] \\
 &=& \frac{1}{n^2}\sum_{t=|r|}^{n-1}\sum_{k=|s|}^{n-1}\esp\left[\epsilon_t \epsilon_k \epsilon_{t-|r|}\epsilon_{k-|s|}\right] \\
 &=& \frac{1}{n^2}\sum_{t=|r|}^{n-1}\sum_{k=|s|}^{n-1}\esp\left[\epsilon_t \epsilon_k \epsilon_{t-|r|}\epsilon_{k-|s|}\right]\1{\{|r| = |s|\}} \label{eq:r_eq_s}\\
 &+& \frac{1}{n^2}\sum_{t=|r|}^{n-1}\sum_{k=|s|}^{n-1}\esp\left[\epsilon_t \epsilon_k \epsilon_{t-|r|}\epsilon_{k-|s|}\right]\1{\{|r| \neq |s|\}}\label{eq:r_neq_s}
\end{eqnarray}
where (\ref{eq:r_eq_s}) is equal to
\begin{eqnarray}
&&\frac{1}{n^2}\left\{\sum_{t=|r|}^{n-1}\sum_{k=|s|}^{n-1}\esp\left[\epsilon_t \epsilon_k \epsilon_{t-|r|}\epsilon_{k-|s|}\right]\1{\{r = s=0\}}
+ \sum_{t=|r|}^{n-1}\sum_{k=|s|}^{n-1}\esp\left[\epsilon_t \epsilon_k \epsilon_{t-|r|}\epsilon_{k-|s|}\right]\1{\{|r| = |s| \neq 0\}}\right\}
\notag \\
&=& \frac{1}{n^2}\left\{\left[n\esp[\epsilon^4_t] + n(n-1)\right]\1{\{r=s=0\}} + (n-|r|)\left(\1{\{r = s \neq 0\}}+ \1{\{r = -s \neq 0\}}\right)\right\}\label{eq:r_eq_s_out}
\end{eqnarray}
and (\ref{eq:r_neq_s}) is equal to
\begin{eqnarray}
&&\frac{1}{n^2}\sum_{t=|r|}^{n-1}\sum_{k=|s|}^{n-1}\esp\left[\epsilon_t \epsilon_k \epsilon_{t-|r|}\epsilon_{k-|s|}\right]\1{\{|r| \neq |s| \neq 0,\; t \neq k\}}\notag \\
&=& \frac{1}{n^2} \sum_{t=|r|}^{n-1}\sum_{k=|s|}^{n-1}\esp[\epsilon_t\epsilon_{k-|s|}]\esp[\epsilon_{t-|r|}\epsilon_k]\1{\{|r| \neq |s| \neq 0,\; t \neq k, \; t-k \neq |r|-|s|\}} \label{eq:t_k_eq_r_s} \\
&+&  \frac{1}{n^2} \sum_{t=|r|}^{n-1}\sum_{k=|s|}^{n-1}\esp[\epsilon_t\epsilon_{t-|r|}\epsilon_k\epsilon_{k-|s|}]\1{\{|r| \neq |s| \neq 0,\; t \neq k, \; t-k= |r|-|s|\}}\label{eq:t_k_neq_r_s}
%&=& \frac{(n-|r|)}{n^2}\esp[\epsilon_t\epsilon_{k-|s|}\epsilon_{t-|r|}\epsilon_k]\1{\{|r| \neq |s| \neq 0,\; t \neq k, \; t-k= |r|-|s|\}} \label{eq:r_neq_s_out}
\end{eqnarray}
For the $\esp[\epsilon_t\epsilon_{k-|s|}]\esp[\epsilon_{t-|r|}\epsilon_k]\neq 0$, we must have $t=k-|s|$ and $t-|r|=k$, which implies that $|s|+|r|=0$.  However, this cannot happen since $\esp[\epsilon_t\epsilon_{k-|s|}]\esp[\epsilon_{t-|r|}\epsilon_k]\1{\{|r| \neq |s| \neq 0,\; t \neq k, \; t-k \neq |r|-|s|\}}$.  Therefore, (\ref{eq:t_k_eq_r_s}) is zero.
Similarly, $\1{\{|r| \neq |s| \neq 0,\; t \neq k, \; t-k= |r|-|s|\}} \neq 0 $ implies that $t=k+|r|-|s|$ and . Thus
\begin{eqnarray}
&\esp[\epsilon_t\epsilon_{t-|r|}\epsilon_k\epsilon_{k-|s|}]\1{\{|r| \neq |s| \neq 0,\; t \neq k, \; t-k= |r|-|s|\}}
= \esp[\epsilon_{k+|r|-|s|}\epsilon_k\epsilon^2_{k-|s|}]\1{\{|r| \neq |s| \neq 0,\; t \neq k, \; t-k= |r|-|s|\}} \notag \\
&= \esp[\epsilon_k]\esp[\epsilon_{k+|r|-|s|}\epsilon^2_{k-|s|}]\1{\{|r| \neq |s| \neq 0,\; t \neq k, \; t-k= |r|-|s|\}}=0
\end{eqnarray}
and (\ref{eq:t_k_neq_r_s}) is zero.
\\
It follows
\begin{equation}
\frac{1}{n^2}\sum_{t=|r|}^{n-1}\sum_{k=|s|}^{n-1}\esp\left[\epsilon_t \epsilon_k \epsilon_{t-|r|}\epsilon_{k-|s|}\right]\1{\{|r| \neq |s| \neq 0,\; t \neq k\}}=0 \;,
\end{equation}
Put (\ref{eq:r_eq_s_out}) into (\ref{eq:pigj_sq}), we obtain
\begin{eqnarray}\label{eq:var_gj}
\esp[{g_j^{\ast}}^2] &=& \esp[(2\pi g_j)^2]-1 \notag \\
&=& \frac{1}{n}\left(\esp[\epsilon^4_t]-1 \right) + \sum_{r \neq 0}\left(\frac{n-|r|}{n^2}\right) k^2\left(\frac{r}{m}\right)+ \sum_{r \neq 0}\left(\frac{n-|r|}{n^2}\right) k^2\left(\frac{r}{m}\right) \rme^{2ir\lambda_j}
\end{eqnarray}
Hence
\begin{eqnarray}\label{eq:norm_esp_gjSqr}
n^{4/5}\esp\left[\frac{1}{\tn}\sum_{j=1}^{\tn}{g_j^{\ast}}^2\right] &=& n^{-1/5}\left(\esp[\epsilon^4_t]-1 \right) + n^{-1/5}\sum_{r \neq 0} k^2\left(\frac{r}{m}\right)
- n^{-6/5}\sum_{r \neq 0}|r| k^2\left(\frac{r}{m}\right) \notag \\
&+& n^{4/5}\sum_{r \neq 0}\left(\frac{n-|r|}{n^2}\right) k^2\left(\frac{r}{m}\right)\left( \frac{1}{\tn}\sum_{j=1}^{\tn}\rme^{2ir\lambda_j}\right)
\end{eqnarray}
where
\begin{equation}\label{eq:convg_ksqr}
n^{-1/5}\sum_{r \neq 0} k\left(\frac{r}{m}\right)^2 = \tau m^{-1}\sum_{r \neq 0} k\left(\frac{r}{m}\right)^2 \rightarrow \tau \int_{\Rset}k^2(x)\rmd x
\end{equation}
\begin{equation}\label{eq:convg_k_x_square}
n^{-6/5}\sum_{r \neq 0}|r| k\left(\frac{r}{m}\right)^2 = \frac{\tau m}{n}\sum_{r \neq 0} \frac{|r|}{m}k\left(\frac{r}{m}\right)^2\frac{1}{m}
\end{equation}
Since
\[
\sum_{r \neq 0} \frac{|r|}{m}k\left(\frac{r}{m}\right)^2\frac{1}{m}\rightarrow \int_{\Rset}|x|k^2(x)\rmd x\;,
\]
Assume $\int_{\Rset}|x|k^2(x)\rmd x < \infty$.
It follows that (\ref{eq:convg_k_x_square}) converges to $0$ as $n \to \infty$.
\\
Next, consider the last term in (\ref{eq:norm_esp_gjSqr}). Note that
\begin{eqnarray*}
\frac{1}{\tn}\sum_{j=1}^{\tn}\rme^{2ir\lambda_j} &=& \frac{1}{\tn}\left[\rme^{i4\pi r/n}\left(\frac{1-\rme^{i4\pi r \tn/n}}{1-\rme^{i4\pi r/n}}\right)   \right]\\
&=& \rme^{i4\pi r/n} \rme^{i2\pi r(\tn-1)/n}\left[\frac{\sin(2\pi r(\tn/n))}{\tn\sin(2 \pi r /n)}   \right] \\
&=& \rme^{i2\pi r(1+\tn)/n}\left[\frac{\sin(2\pi r(\tn/n))}{\tn\sin(2 \pi r /n)} \right]
\end{eqnarray*}
Hence
\begin{eqnarray}
&& n^{4/5}\sum_{r \neq 0}\left(\frac{n-|r|}{n^2}\right) k\left(\frac{r}{m}\right)^2 \frac{1}{\tn}\sum_{j=1}^{\tn}\rme^{2ir\lambda_j}\notag\\
&=& n^{4/5}\sum_{r \neq 0}\left(\frac{n-|r|}{n^2}\right) k\left(\frac{r}{m}\right)^2\rme^{i2\pi r(\tn+1)/n}\left[\frac{\sin(2\pi r(\tn/n))}{\tn\sin(2 \pi r /n)}   \right] \label{eq:eva_all}\\
&=& n^{-1/5}\sum_{r \neq 0}\left(1-\frac{|r|}{n}\right) k\left(\frac{r}{m}\right)^2\rme^{i2\pi r(\tn+1)/n}\left[\frac{\sin(2\pi r(\tn/n))}{\tn\sin(2 \pi r /n)}   \right] \label{eq:eva_1}
\end{eqnarray}
\\
Let $\omega_{2r} = 4 \pi r / n$. Denote $D_{\tn}(\omega_{2r})=\frac{\sin(2\pi r(\tn/n))}{\tn\sin(2 \pi r /n)}$.
Note that $0 \leq k(r/m) \leq 1$, if $0 < |r| \leq m$; $k(|r|/m)=0$, if $|r| > m$; $k(r/m)=k(-r/m)$.
\\
It follows that
\begin{eqnarray}
&& \left| n^{-1/5}\sum_{r \neq 0} \left(1-\frac{|r|}{n}\right)k\left(\frac{r}{m}\right)^2\rme^{i2\pi r(\tn+1)/n}\left[\frac{\sin(2\pi r(\tn/n))}{\tn\sin(2 \pi r /n)} \right] \right| \notag \\
& = & \left| n^{-1/5}\sum_{0 < |r| \leq m}\left(1-\frac{|r|}{n}\right) k\left(\frac{r}{m}\right)^2\rme^{i(\tn+1)\omega_{2r}}D_{\tn}(\omega_{2r})\right| \notag \\
& \leq & n^{-1/5}\sum_{0 < |r| \leq m} k\left(\frac{r}{m}\right)^2 |D_{\tn}(\omega_{2r})|  \notag \\
& \leq & n^{-1/5}\sum_{0 < |r| \leq m}|D_{\tn}(\omega_{2r})|
\label{eq:eva_2}
\end{eqnarray}
By Corollary \ref{cor:conv_fft_lag_win}, for every $|r| \leq m$, if $\omega_{2r} \in (-\pi, \pi)$, where $n \geq (4 \tau)^{5/4}$, and if $c > 4/5$,
(\ref{eq:eva_2}) converges to 0.
\\
In view of the convergence results of (\ref{eq:norm_esp_gjSqr}), (\ref{eq:convg_ksqr}), (\ref{eq:convg_k_x_square}), and (\ref{eq:eva_2}),  we prove (\ref{eq:convg_exp_g}).

\end{proof}

\end{document}